\documentclass[a4paper, twocolumn, 10pt, accepted=2023-10-20]{quantumarticle}
\pdfoutput=1
\usepackage[utf8]{inputenc}
\usepackage[T1]{fontenc}
\usepackage[english]{babel}
\usepackage{csquotes}\MakeOuterQuote"
\usepackage[numbers,sort&compress]{natbib}
\usepackage{amsmath,amsfonts,amssymb,amsthm,bbm}
\usepackage[colorlinks]{hyperref}
\usepackage[capitalize]{cleveref}

 \newtheorem{theorem}{Theorem}
 \newtheorem{proposition}{Proposition}

 \newtheorem{lemma}[theorem]{Lemma}
 
 \newtheorem{definition}[theorem]{Definition}

\newcommand{\mc}[1]{\mathcal{#1}}

\newcommand{\mb}[1]{\mathbb{#1}}

\newcommand{\tr}{\mathrm{Tr}} 
\newcommand{\Tr}{\mathrm{Tr}} 

\newcommand{\id}{\mathbbm{1}}

\newcommand{\RR}{\mb{R}}

\newcommand{\norm}[1]{\left\Vert #1 \right\Vert}

\newcommand{\ket}[1]{|{#1}\rangle}

\newcommand{\bra}[1]{\langle{#1}|}

\newcommand{\braket}[2]{\langle #1 | #2 \rangle}

\newcommand{\ketbra}[2]{\ket{#1} \bra{#2}}

  \newcommand{\proj}[1]{\ketbra{#1}{#1}}

\def\id{{\mathbbm 1}}
\newcommand{\Id}{\mathrm{Id}}
\newcommand{\NN}{\mathbb{N}}

\begin{document}
\author{Lauritz van Luijk}
\affiliation{Leibniz Universit\"at Hannover, Appelstra\ss e 2, 30167 Hannover, Germany}
\orcid{0000-0003-3153-549X}
\author{Reinhard F. Werner}
\affiliation{Leibniz Universit\"at Hannover, Appelstra\ss e 2, 30167 Hannover, Germany}
\orcid{0000-0003-2288-468X}
\author{Henrik Wilming}\affiliation{Leibniz Universit\"at Hannover, Appelstra\ss e 2, 30167 Hannover, Germany}
\orcid{0000-0002-0306-7679}
\title{Covariant catalysis requires correlations and good quantum reference frames degrade little}

\begin{abstract}
Catalysts are quantum systems that open up dynamical pathways between quantum states which are otherwise inaccessible under a given set of operational restrictions while, at the same time, they do not change their quantum state.
We here consider the restrictions imposed by symmetries and conservation laws, where any quantum channel has to be covariant with respect to the unitary representation of a symmetry group, and present two results.
First, for an exact catalyst to be useful, it has to build up correlations to either the system of interest or the degrees of freedom dilating the given process to covariant unitary dynamics. 
This explains why catalysts in pure states are useless. Second, if a quantum system ("reference frame") is used to simulate to high precision unitary dynamics (which possibly violates the conservation law) on another system via a global, covariant quantum channel, then this channel can be chosen so that the reference frame is approximately catalytic.
In other words, a reference frame that simulates unitary dynamics to high precision degrades only very little. 
\end{abstract}
\maketitle

\section{Introduction}
When a quantum system $C$ takes part in a physical interaction with a different system $S$, its quantum state typically changes.
Moreover, the system $C$ typically becomes correlated to system $S$ (if the two weren't correlated, to begin with). 
Indeed, suppose contrarily that the joint time evolution of $SC$ is given by a unitary operator $U$ and the two systems end up uncorrelated, and the state of $C$ did not change:
\begin{align}
U \rho_S\otimes \sigma_C U^\dagger  = \rho'_S \otimes \sigma_C.
\end{align}
Then there also exists a unitary $V$ such that $V\rho_S V^\dagger=\rho'_S$.
Thus, the system $C$ is not required to realize the state-transition $\rho_S\rightarrow \rho'_S$ on $S$.

Once the operations we can implement in an experiment are restricted to some strict subset $\mc O$ of all conceivable quantum operations, however, there may be situations where the state-transition $\rho_S\rightarrow \rho_S'$ is \emph{impossible} via operations in $\mc O$, but 
the state-transition $\rho_S\otimes\sigma_C \rightarrow \rho_S'\otimes \sigma_C$ \emph{is possible} via operations in $\mc O$.
Thus, the mere presence of system $C$ seems to allow for larger sets of state transitions on $S$. 
The prototypical example is given by the set $\mc O_{\mathrm{LOCC}}$ of \emph{local operations and classical communication} (LOCC), where $S=AB$ and $C=C_AC_B$ each consist of two parts distributed between two players, Alice and Bob. 
Alice and Bob each can only perform quantum operations on their local systems but can communicate classically, for example, by a phone call. 
It was observed in Ref.~\cite{Jonathan1999} that for some pure quantum states $\ket{\psi}_{AB},\ket{\psi'}_{AB}$ and $\ket{\phi}_{C}$ the state-transition $\ket{\psi}_{AB}\otimes\ket{\phi}_C \rightarrow \ket{\psi'}_{AB}\otimes\ket{\phi}_C$ is possible even though $\ket{\psi}_{AB}\rightarrow \ket{\psi'}_{AB}$ is impossible via LOCC protocols. 
Ref.~\cite{Eisert2000} generalized this observation to mixed states. 
In such a situation, system $C$ plays the role of a \emph{catalyst} since it enables a state transition without changing its state itself.
In particular, it may be re-used to implement the same state transition on a new system $S'=A'B'$. 

In this work, we consider the role of catalysts when the operational restrictions are due to a set of conservation laws that need to be obeyed \cite{Janzing2006,Bartlett2007,Vaccaro2008,Gour2008}. 
As a consequence (see below), every quantum system $S$ carries a unitary representation $g\mapsto W_S(g)$ of a connected Lie group $G$. 
The set of implementable operations, which we denote by $\mc O_{\mathrm{symm.}}$, corresponds to \emph{covariant} quantum channels $T$. 
A quantum channel from system $S$ to system $S'$ is said to be covariant if it fulfills
\begin{align}
	T[W_S(g)\rho W_S(g)^\dagger] = W_{S'}(g)T[\rho]W_{S'}(g)^\dagger.	
\end{align}
A quantum state $\rho$ on $S$ is called \emph{symmetric} if $\rho = W_S(g)\rho W_S(g)^\dagger$ for all $g\in G$. A covariant quantum channel always maps symmetric states to symmetric states. 
Since all the operations we can implement are covariant, we can only prepare symmetric states. 
In such a setting, states that are not symmetric are valuable resources, as they can be used to implement
quantum channels that are not covariant: If $\omega_E$ is an asymmetric state on system $E$, and $T$ is a covariant quantum channel on $SE$, then
the channel
\begin{align}
	\rho\mapsto \tr_E[T[\rho\otimes\omega_E]]
\end{align}
is, in general, not covariant. 
In particular, a useful catalyst must be asymmetric; otherwise, the induced dynamics on $S$ were covariant and would not require the catalyst.

In this context, asymmetric states are also called \emph{quantum reference frames}, since they break the underlying symmetry and thereby allow to (partially) distinguish symmetry operations: 
If $\rho$ is asymmetric, then at least for some group elements $g\in G$ the state $\rho$ is (partially) distinguishable from $W_S(g)\rho W_S(g)^\dagger$.
An example of a reference frame for rotations is a collection of spins in strongly polarized states. 
A superposition (in contrast to an incoherent mixture) of energy eigenstates provides a reference frame for the group of time translations (a phase-reference in the case of equally spaced energy levels). 
A \emph{perfect} quantum reference corresponds to a quantum state $\rho$ that is perfectly distinguishable from $W_S(g)\rho W_S(g)^\dagger$, i.e.\ the trace-distance
\begin{align}
    D(W_S(g)\rho W_S(g)^\dagger,\rho) := \frac{1}{2}\norm{W_S(g)\rho W_S(g)^\dagger - \rho}_1
\end{align}
is equal to $1$ unless $g=1$. Perfect reference frames cannot exist for continuous groups:
If $W_S$ is a strongly continuous (possibly projective) unitary representation on a (possibly infinite-dimensional) Hilbert-space, then $g\mapsto D(W_S(g)\rho W_S(g)^\dagger,\rho)$ is a continuous function on $G$ and hence cannot be equal to the discontinuous map $g\mapsto 1- \delta_{1g}$  
\footnote{This argument even works in operator algebraic settings: If the system $S$ is described by a von Neumann algebra $\mathcal M$ with $\alpha:G \to \mathrm{Aut}(\mathcal M)$ a (normal) action of $G$, then $\norm{\alpha_g(\rho)-\rho}_{\mc M^*}$ is continuous for all normal states $\rho\in \mc M_*$.}.

This paper aims to show two results about catalysts in the context of conservation laws. 
First, we show that a useful catalyst must build up correlations to either the system of interest $S$ or the environmental degrees of freedom $E$ (or both) that dilate the covariant channel on $SC$ to a covariant unitary channel. 
In particular, if the process on $SC$ is already unitary, then $C$ must become correlated to $S$.
Since a quantum system in a pure state is always uncorrelated to any other system, catalysts must be mixed to be useful. This is in strong contrast to the setting of LOCC.
Our result generalizes previous results showing that pure catalysts are useless in the case where only energy is preserved \cite{Ding2021} or when the states of interest are pure \cite{Marvian2013}.
We discuss the role of correlations in catalysis further below. 

Our second result concerns the case where a quantum reference frame $C$ is used to locally implement approximately unitary dynamics on $S$ that may violate the conservation laws. 
Such a situation is quite common in a semi-classical limit, where, for example, a laser (described by a classical field) is used to induce a state transition of an atom from the ground state to the superposition of the ground state and an excited state, violating energy conservation.  
It is already known that a quantum reference frame that allows to locally implement approximately unitary dynamics must be much larger than $S$ and strongly asymmetric \cite{Ozawa2002,Bartlett2007,Gour2009,Marvian2012,Marvian2012a,Marvian2013,Aaberg2014,Miyadera2016,Miyadera2020,Tajima2018,Tajima2020,Marvian2021,Tajima2021,Tajima2022}.
In this sense it must be close to macroscopic.  
For example, a laser described by a (quantum) coherent state with high expected photon number is of the right form. 
It is also known that a quantum reference frame generically degrades in quality when it is used due to back-action \cite{Bartlett2006,Bartlett2007a,Poulin2007,Ahmadi2010}.
We show that a quantum reference frame that allows implementation of local unitary dynamics to high precision can always be used so that it is approximately catalytic: its quantum state is only slightly perturbed. 
Therefore, such a quantum reference frame is also macroscopic in the sense of receiving little back-action.
In particular, if it is initially in a pure state, it only becomes weakly correlated to $S$. 
The second result holds for covariant operations with respect to projective unitary representation of any group, not just connected Lie groups.

The paper is structured as follows: 
In section~\ref{sec:resource_theories} we provide the necessary background on the framework that we work in, namely
the resource theory of asymmetry. 
We also introduce the closely related resource theory of athermality as our first main result immediately transfers to that setting as well. 
In section~\ref{sec:correlations} we then discuss the role of correlations for catalysis and present our first main result and its proof, which builds on Wiegmann's theorem in matrix theory.
In section~\ref{sec:embezzlement} we discuss how quantum reference frames can be used to locally circumvent conservation laws and present our second main result together with its proof, which relies on the information-disturbance tradeoff in quantum mechanics.

\section{The resource theories of asymmetry and athermality}
\label{sec:resource_theories}

We now introduce the resource theory of asymmetry and then also the resource theory of athermality, which provides a model for thermodynamics. For more background on the resource theory of asymmetry, see Refs.~\cite{Bartlett2007,Gour2008,Marvian2012,Marvian2013} and references therein and for more background on the resource theory of athermality, see Refs.~\cite{Janzing2000,Horodecki2013,Brandao2013,Brandao2015,Gour2015,Faist2015,Gallego2016,YungerHalpern2016}.

The resource theory of asymmetry formalizes the consequences of symmetries and conservation laws 
for the manipulation of quantum systems. 

Symmetries are implemented in quantum physics by projective unitary representations of a symmetry group $G$  on Hilbert-space \footnote{In fact, to be fully general, one should also consider anti-unitary operators, but here we restrict to (projective) unitary representations.}. 
It is then required that the Hamiltonian generating the time-evolution commutes with this representation.
If the symmetry group $G$ is a Lie group, then its Lie algebra $\mathfrak{g}$ is the set of infinitesimal generators of the group. We can view an infinitesimal generator $X$ as an abstract conserved quantity because its representative $X_S$ on a system $S$ is a Hermitian operator commuting with the system Hamiltonian $H_S$.
Notationally, we distinguish an element $X$ of the Lie algebra and its Hermitian representative $X_S$ on a system $S$ only by the label, so that $W_S(\mathrm e^{X})=\mathrm{e}^{-\mathrm i X_S}$.

In the resource theory of asymmetry with respect to a group $G$, it is assumed that every system $S$ carries a projective unitary representation $g\mapsto W_S(g)$ of $G$ commuting with its Hamiltonian $H_S$, with the consistency condition that any two separate systems, $S_1$ and $S_2$, carry jointly the tensor-product representation $g\mapsto W_{S_1}(g)\otimes W_{S_2}(g)$ when considered as a single system $S=S_1S_2$. 
Then any unitary $U$ realizing a physical operation on a (possibly multi-partite) system $S$ must commute with the representation:
\begin{align}
	[U, W_S(g)]=0\quad \forall g\in G.
\end{align}
In the case of a connected Lie group, this is equivalent to
\begin{align}
	[U,X_S]=0\quad \forall X\in \mathfrak g.
\end{align}
Therefore, all $X_S$, with $X\in \mathfrak g$, are indeed conserved quantities.
In the case of finite-dimensional Hilbert-spaces and for a connected Lie group $G$, we can, without loss of generality, assume that all systems are equipped with proper unitary representations \cite{bargmann}. 
For this, we pass to the universal covering group of $G$, i.e.\ the unique simply-connected group generated by the Lie algebra $\mathfrak{g}$.
For a compound system $S=S_1S_2$ we have
\begin{align}
    X_S=	X_{S_1}\otimes \id_{S_2} + \id_{S_1} \otimes X_{S_2} \quad X\in \mathfrak g.
\end{align}

We may further enrich the set of possible operations $\mc O_{\mathrm{symm}}$ by assuming that we can prepare "for free" arbitrary systems $E$ in an arbitrary \emph{symmetric density matrix} $\omega_E$, which fulfills $[\omega_E,W_E(g)]=0$ for all $g\in G$. 
Symmetric states may be interpreted as those which can be prepared without access to a reference frame for the group \cite{Bartlett2007}.
By further being able to discard subsystems (formally tracing them out), it is then possible to implement any quantum channel $T$ between two systems $S$ and $S'$ that is \emph{covariant}:
\begin{align}
	T[W_S(g) \rho_S W_S(g)^\dagger] &= W_{S'}(g) T[\rho_S] W_{S'}(g)^\dagger \quad \forall g\in G. 
\end{align}
Indeed, the \emph{covariant Stinespring theorem} states that any covariant quantum channel $T$ can be written as
\begin{align}
	T[\rho] = \tr_E[U \rho\otimes \omega_E U^\dagger],
\end{align}
where $\omega_E$ is a symmetric state that, in fact, may be assumed to be pure \cite{Keyl1999}. 
The preparation of a symmetric state on $S$ can be seen as a covariant quantum channel from a trivial system to $S$.

Summarizing, the resource theory of asymmetry consists of quantum channels that are covariant with respect to the (projective) unitary representations of a symmetry group $G$.
As already mentioned in the introduction, any quantum state that is not symmetric can be understood as a reference frame for the group $G$ and may be utilized to implement non-covariant quantum channels via a covariant unitary $U$:
If $[\omega_E,W_E(G)]\neq 0$, then in general the quantum channel
\begin{align}
	\rho\mapsto \tr_2[U\rho\otimes\omega_E U^\dagger]
\end{align}
is not covariant, even though $U$ commutes with $W_S(g)\otimes W_E(g)$ for all $g\in G$.

Let us now discuss the resource theory of athermality. It is closely related to the resource theory of asymmetry for the group $G=\mathbb R$ of time translations,
which enforces strict energy conservation in the sense that any unitary operation has to commute with the full Hamiltonian.
In this case, the group representation is simply given by unitary time-evolution with the system's Hamiltonian, $t\mapsto W_S(t) = \exp(-\mathrm it H_S)$. 
In contrast to the resource theory of asymmetry where arbitrary symmetric states can be prepared "for free", the resource theory of athermality only allows for the Gibbs states at some fixed inverse temperature $\beta >0$, i.e.\ states of the form $\omega_S = \exp(-\beta H_S)/Z$ with $Z=\Tr[\exp(-\beta H_S)]$, instead of arbitrary symmetric states.
The underlying idea is that such systems represent heat baths at some fixed environment temperature. 
Therefore, we also only allow for those covariant quantum channels $T$ that can be realized by appending a Gibbs state of an arbitrary Hamiltonian $H_E$, applying an energy-conserving unitary and tracing out a subsystem:
\begin{align}
    T[\rho] = \tr_E[U \rho\otimes\omega_E U^\dagger],
\end{align}
where $\omega_E=\exp(-\beta H_E)/Z$ and $[U,H_S\otimes\id_E + \id_S\otimes H_E]=0$.
We refer to the resulting quantum channels as "thermal operations" \cite{Horodecki2013}.
Any thermal operation has the Gibbs state at the fixed inverse temperature $\beta$ as a fixed-point: $T[\exp(-\beta H_S)] = \exp(-\beta H_S)$. 
This can be seen as enforcing the second law of thermodynamics: without spending resources, we cannot produce a non-equilibrium state from thermal equilibrium. 
Any unitary quantum channel that is covariant with respect to time translations automatically has the Gibbs state as a fixed point and is a valid quantum channel in this framework.

\section{Correlations and catalysis}
\label{sec:correlations}
We now discuss the role of correlations for catalysis and discuss our first main result. 
Traditionally, catalysis in the context of quantum mechanics, in particular LOCC, always referred to a situation where $\rho_S\rightarrow \rho'_S$ is impossible under the given set of operations $\mc O$, but $\rho_S\otimes\sigma_C\rightarrow \rho'_S\otimes \sigma_C$ is possible.
In other words, no correlations between $S$ and $C$ were allowed.
This assumption can be justified by considering several catalytic transformations. For example, assume $S$ interacts with $C_1$ in such a way that $S$ and $C_1$ became correlated, but $C_1$ did not change its state. Afterwards, $S$ interacts with $C_2$ in a similar fashion.
Then it will typically be the case that $C_1$ and $C_2$ become correlated. Therefore, while each is catalytic on their own, the joint-system $C_1C_2$ is not catalytic.
Similarly, if a second system $S'$ interacts with $C_1$ in the same way as $S$ did, then afterwards, the reduced density matrix of $S'$ will be precisely the same as that of $S$. However, the two will typically be correlated \cite{Vaccaro2018,Boes2019}.
In other words, once we allow for the build-up of correlations, these correlations will typically spread in an uncontrolled way (the only way to completely prevent this is to demand that the states of catalysts are pure).

However, assume now that the interaction of $S$ and $C_1$ is not unitary so that auxiliary degrees of freedom $E$ are involved, as is typically the case in a resource theory. 
Then even though these degrees of freedom may be considered "for free", they can similarly result in uncontrolled spreading of correlations if $C_1$ (or $S$ for that matter) becomes correlated to them. 
Hence, if one does not restrict correlations with $E$, why restrict those between $S$ and $C$?
Moreover, $E$ is typically treated implicitly so that the correlations to $E$ do not appear in the mathematical treatment. 

On the other hand, it has recently been shown in a number of works that it can be quite beneficial to allow for correlations between catalyst and system as they can significantly enhance the set of reachable states from a given initial state \cite{Mueller2018,Rethinasamy2020,Shiraishi2021,Boes2019,Wilming2021,LipkaBartosik2021,Kondra2021,Wilming2022}, even if no auxiliary degrees of freedom $E$ are allowed for \cite{Boes2019,Wilming2021,Wilming2022}. 
In fact, oftentimes such \emph{correlated catalysis} expands the set of reachable states for a single system to the asymptotic set of reachable states in the thermodynamic limit: Suppose for any $\epsilon>0$ there exists a number $n$ such that the state transition $\rho^{\otimes n}\rightarrow \rho'_{n,\epsilon}$ is possible without a catalyst, where $\rho'_{n,\epsilon}$ approximates ${\rho'}^{\otimes n}$ up to error $\epsilon$ in trace-distance. 
It is quite common that such a situation occurs even though the transition $\rho\rightarrow \rho'$ is impossible with good accuracy.
The asymptotic limit of many uncorrelated copies therefore strongly regularizes possible state-transitions.
The same regularization can be achieved with correlated catalysis: In the setting described above, for any $\epsilon>0$ a state-transition $\rho\rightarrow \rho'_\epsilon$ is possible with correlated catalysis, where $\rho'_\epsilon$ is $\epsilon$-close to $\rho'$ in trace-distance.
Moreover, often these beneficial effects are possible while building up arbitrary little correlations between $S$ and $C$ as measured by the trace distance to the closes product state.
At the same time, it is typically unknown how strongly $C$ becomes correlated to $E$.

These considerations motivate our first main result. 
It shows that a catalyst in the resource theories of asymmetry and athermality can only be useful if it either becomes correlated to $S$ or $E$ (or both).

\begin{theorem}[Correlations are neccessary for catalysis]\label{thm:finite}
	Let $G$ be a connected Lie group and let $g\mapsto W_S(g)$, $g\mapsto W'_S(g)$ and $g\mapsto W_C(g)$ be finite-dimensional unitary representations on systems $S$ and $C$, respectively.
	Let $\rho_S,\rho'_S$ be density matrices on $S$ and $\sigma_C$ a density matrix on $C$. Finally, suppose a unitary $U$ on $SC$ fulfills
	\begin{align}
		U\rho_S\otimes \sigma_C U^\dagger = \rho'_S\otimes \sigma_C\label{eq:thm1}
	\end{align}
	and for all $g\in G$ 
	\begin{align}
		U W_S(g)\otimes W_C(g) = W_S'(g)\otimes W_C(g) U. \label{eq:thm2}
	\end{align}
	Then there exists a unitary $V$ on $S$ such that $V\rho_S V^\dagger  = \rho'_S$ and $V W_S(g) = W'_S(g) V$ for all $g\in G$.
\end{theorem}
The non-trivial statement of the theorem is that the unitary $V$ can also be chosen to intertwine the two group representations $W_S$ and $W'_S$.
The theorem can be interpreted most clearly when we substitute the system $S$ in the statement of the theorem by a bipartite system $SE$,
where $E$ contains the degrees of freedom required to dilate a covariant quantum channel $T$ on $SC$ to a unitary channel:
\begin{align}
	T[\rho_S\otimes \sigma_C]=\tr_E[U \rho_S\otimes\omega_E\otimes\sigma_C U^\dagger].
\end{align}
Here, $\omega_E$ is a general symmetric state in the resource theory of asymmetry and a Gibbs state in the resource theory of athermality. Theorem~\ref{thm:finite} then states that if a system $C$ takes part in a process but neither experiences back-action (does not change its density matrix $\sigma_C$) nor becomes correlated to $SE$, then 
the system $C$ was not necessary to realize the state transition on $S$ (and $E$). This is because the covariant quantum channel $T'$ defined by
\begin{align}
	T'[\rho_S] = \tr_E[V \rho_S\otimes\omega_E V^\dagger]
\end{align}
has the same effect on $S$. Moreover, $T'$ is also a thermal operation if $T$ was.
In other words: For a catalyst to be useful in realizing otherwise impossible state transitions, it must become correlated
to either $S$, $E$, or both, and this statement holds both in the resource theory of asymmetry for connected Lie groups and in the resource theory of athermality. Note that the usual setting of uncorrelated catalysis allows for correlations between $S$ and $E$ as well as between $C$ and $E$ but not between $C$ and $S$.

Therefore, a catalyst in a pure state is useless in contrast to the setting of LOCC.
While this statement is known in the cases of either a single conserved quantity (i.e.\ $G=\RR$) \cite{Ding2021} or the case where $\rho_S$ is pure \cite{Marvian2013}, our result shows that it holds for any finite number of conserved quantities and any state $\rho_S$.
More importantly, it explains why this is true: A catalyst in a pure state cannot be correlated to another system.
If we interpret $C$ as a reference frame, this means that a reference frame that is actually used (even only once) must either degrade or become correlated to other systems.

One may wonder whether there actually are situations where it is helpful to have a system $C$ that only becomes correlated to $S$ but does not experience back-action. 
This is indeed the case: There are state transitions $\rho_S\rightarrow \rho'_S$ that can be realized in this way, but that cannot be realized using only covariant operations \cite{Ding2021}.
However, it is impossible to transform a symmetric state to an asymmetric state in this way, as shown by so-called no-broadcasting theorems \cite{Marvian2019,Lostaglio2019}.
One can also construct examples in the resource theory of athermality where an asymmetric (coherent) catalyst that becomes correlated to $S$ allows for state transitions that are impossible to implement without a catalyst \cite{Patryk}.

It is important to emphasize that our results hold for one particular choice of density matrices: We do \emph{not} require that $\sigma_C$ remains unchanged and uncorrelated for \emph{all} possible density matrices on $S$, but only for the specific density matrix $\rho_S$. 
We have mentioned above that a system $C$ in a quantum state that is not symmetric with respect to a group representation can be viewed as an (imperfect) reference frame for the group:
By interacting with it using only covariant quantum channels, one can implement some non-covariant quantum channels on a system $S$.
It is already known that in this case, for at least some state on $S$ (including the maximally mixed state), the final state on $C$ cannot be brought back exactly to its initial state by a covariant quantum channel \cite{Marvian2016}. 
This means that the reference frame "degrades" for such inputs. 
Conversely, if the reference frame does not degrade for any input on $S$, then the dynamics on $S$ is already covariant. 
The degradation of quantum reference frames has been studied in several concrete examples, see Refs.~\cite{Bartlett2006,Bartlett2007a,Poulin2007,Ahmadi2010}.

Theorem \ref{thm:finite} then can also be interpreted as showing that whenever one uses a quantum reference beneficially, it must either experience some back-action (change its density matrix) or become correlated to the system of interest $S$. 
In particular, back-action is unavoidable if the reference is initially in a pure state.

In section~\ref{sec:embezzlement} we will see a counter-point to this result: if a quantum reference frame is of such high quality that it allows implementing approximately unitary dynamics on $S$, then it suffers only mild back-action.

The amount of correlations that the catalyst $C$ establishes with the system and environment $SE$ can be measured in terms of the \emph{mutual information} $I(C:SE)$. 
From the unitary invariance of the von Neumann entropy $H$ it then follows that
\begin{align}
	I(C:SE)_{U \rho_{SE}\otimes \sigma_C U^\dagger} = H(\rho'_{SE}) - H(\rho_{SE}).
\end{align}
Thus, in the context of the resource theory of asymmetry, a catalyst is only useful if the entropy on $SE$ strictly increases.
Similarly, one can check that the rank of $\rho'_{SE}$ cannot be smaller than that of $\rho_{SE}$. 
Recently, it was shown in Refs.~\cite{Boes2019,Wilming2021,Wilming2022} that if $\rho_S$ and $\rho'_S$ are two finite-dimensional density matrices such that $H(\rho_S)<H(\rho'_S)$ and $\mathrm{rank}(\rho_S)\leq \mathrm{rank}(\rho'_S)$,
then there exists a finite-dimensional density matrix $\sigma_C$ and a unitary $U$ such that 
\begin{align}
	\tr_2[U\rho_S\otimes \sigma_C U^\dagger] = \rho'_S,\quad \tr_1[U\rho_S\otimes\sigma_C U^\dagger]=\sigma_C.
\end{align}
However, as far as we are aware, it is currently an open problem to decide when $U$ and $\sigma_C$ exist under the additional requirement
that $C$ carries a unitary representation of a given group $G$, and $U$ intertwines the respective representations as in Theorem~\ref{thm:finite} (see Refs.~\cite{Marvian2013,Gour2018,Alexander2022} for progress on the general state-convertibility problem in the resource theory of asymmetry).

The assumption of a {\it connected} Lie group in the theorem is necessary. The theorem fails for any non-connected Lie group $G$:
Let $G$ be a Lie group with finitely many connected components, and let $H$ be the quotient group $G/G_0$ w.r.t.\ the connected component of the neutral element (which is always a normal subgroup). 
Note that $H$ is a finite group and consider the Hilbert space $\mc H_C=\ell^2(H)$. 
We may view the left regular representation of $H$, given by $W_C(x) \ket y =\ket{xy}$, $x,y\in H$, as a unitary representation of $G$ (with kernel equal to $G_0$). 
Given states, $\rho$ and $\sigma$ on any system $S$ which carries a representation $W_S$ of $G$ with $G_0 \subset \ker W_S$, i.e., a system where $W_S$ is essentially a representation of $H$, one can always find a covariant quantum channel $\mc E$ on $SC$ which maps $\rho\otimes\ketbra xx$ to $\sigma\otimes\ketbra xx$ for some arbitrary but fixed $x\in H$. For example, this channel could be given by \cite{Marvian2013}
\begin{multline}
    \mc E[X] = \sum_{y\in H} \tr[(\id\otimes \ketbra yy) X] \times \ldots\\[-5pt]
     \times W_C(yx^{-1}) \sigma W_C(x^{-1}y)\otimes\ketbra yy.
\end{multline}
The reason why such a construction is possible is that the states $\ket x$ for $x\in H$ are perfect reference frames for $H$: $\braket{x}{y} = \delta_{xy}$ for all $y\in H$.
For general finite groups $G$, any quantum channel on $S$ can be realized by a covariant quantum channel on $SC$ if $C$ carries the left regular representation of $G$. 
Concretely, let $T[\cdot]=\tr_E[V\sigma_E\otimes(\cdot) V^\dagger]$ be the channel to be implemented. Define $V_y := (\id_E\otimes W_S(y))V(\id_E \otimes W_S(y)^\dagger)$. Then the covariant channel
\begin{align}
    \mc E(X) = \sum_y \tr_{E}[(V_y \otimes \bra y) \sigma_E\otimes X (V_y^\dagger\otimes\ket y)]\otimes \proj{y}\nonumber
\end{align}
fulfills
\begin{align}
    \mc E[\rho\otimes\proj{1}] = T[\rho]\otimes\proj 1.
\end{align}

\subsection{The proof}

We will now prove Theorem~\ref{thm:finite}.
The important observation is that since we consider connected Lie groups $G$, the infinitely many equations \eqref{eq:thm2} reduce to a finite set of equations. 
Namely, let $\{X^{(i)}\}_{i=1}^{\mathrm{dim}(G)}$ be a (Hermitian) basis for the Lie algebra of $G$ (corresponding to conserved quantities).
The $X^{(i)}$ are represented on the systems $S$ and $C$ by the given representations as $X^{(i)}_S$, $Y^{(i)}_S$ and $X^{(i)}_C$, respectively.
(That is, $Y^{(i)}_S$ represents $X^{(i)}$ on system $S$ as induced by the representation $W'_S$.)
We also introduce the strictly positive matrices
\begin{align}
	\omega_S^{(i)}\otimes \omega_C^{(i)} &:= \exp(-X^{(i)}_S)\otimes\exp(-X^{(i)}_C) ,  \nonumber\\
	\hat \omega_S^{(i)}\otimes\omega_C^{(i)} &:= \exp(-Y^{(i)}_S)\otimes\exp(-X^{(i)}_C).  \nonumber
\end{align}
Up to normalization, we may interpret the operators $\omega_{S/C}^{(i)}$ as Gibbs states of the associated conserved quantities at temperature $T=1$. The eqs.~\eqref{eq:thm1} and \eqref{eq:thm2} are equivalent to the $\mathrm{dim}(G)+1$ equations 
\begin{align}
	U \rho\otimes \sigma U^\dagger &= \rho'\otimes \sigma\\
	U \omega_S^{(i)}\otimes \omega_C^{(i)} U^\dagger &= \hat\omega_S^{(i)}\otimes \omega_C^{(i)}.
\end{align}
To arrive at this result, simply observe that $[U,X]=0$ if and only if $[U,\exp(X)]=0$ for a unitary $U$ and a Hermitian matrix $X$.
Thus, we have a system of $m+1=\mathrm{dim}(G)+1$ equations of the form
\begin{align}
	U A^{(i)} \otimes C^{(i)} U^\dagger = B^{(i)}\otimes C^{(i)},
\end{align}
where $A^{(i)}$, $B^{(i)}$ and $C^{(i)}$ are positive operators.

Theorem \ref{thm:finite} is now an immediate consequence of the following proposition:
\begin{proposition}\label{prop:main}
	Let $\{A_i = A^{(i)}\otimes C^{(i)}\}_{i=0}^m$ and $\{B_i = B^{(i)}\otimes C^{(i)}\}_{i=0}^m$ be two sets of positive semi-definite matrices such that all matrices with $i\geq 1$ have full rank.
	Then if $UA_i U^\dagger =B_i$ for a unitary $U$ there exists a unitary $V$ such that $V A^{(i)} V^\dagger = B^{(i)}$.
\end{proposition}
The main technical tool we will make use of is Wiegmann's theorem \cite{Wiegmann1961}, which is a generalization of Specht's theorem \cite{Specht1940} (see the survey by Shapiro \cite{Shapiro1991} for more results on unitary invariance). 
To state it, we define a \emph{word} $w(x_0,\ldots,x_m)$ as a monomial in the $m+1$ non-commuting variables $x_i$. 
For example, one possible word for $m=2$ is given by $x_0^0 x_2^2 x_1^1 x_0^4$ (we do not assume that $x^0=1$).
\begin{theorem}[Wiegmann \cite{Wiegmann1961}]\label{thm:wiegmann} Let $(A_0,\ldots,A_m)$ and $(B_0,\ldots,B_m)$ be two tuples of complex $d\times d$-matrices. Then there exists a unitary matrix $U$ such that
	$U A_i U^\dagger = B_i$ for all $i=0,\ldots,m$ if and only if
	\begin{align}
		\Tr&[w(A_0,A_0^\dagger,A_1,A_1^\dagger,\ldots,A_m,A_m^\dagger)]\nonumber \\
		\quad&=\Tr[w(B_0,B_0^\dagger,B_1,B_1^\dagger,\ldots,B_m,B_m^\dagger)]
	\end{align}
	for every word $w$ in $2(m+1)$ variables. 
\end{theorem}
Since only positive operators appear in our case, we restrict ourselves to words in $m+1$ variables. Now we fix an arbitrary word $w(x_0\dots,x_m)$ of $(m+1)$ variables and show that the conditions of the proposition imply
\begin{equation}
\begin{aligned}
	\Tr[w(A^{(0)},\ldots,A^{(m)})] = \Tr[w(B^{(0)},\ldots,B^{(m)})].
\end{aligned}
\end{equation}
Then Wiegmann's theorem implies that the required unitary $V$ exists. 

We now define the length $L$ of the word $w$ as follows: We write out the word in the form $\ell_1^{n_1}\cdots \ell_L^{n_L}$, where each letter $\ell_j$ is taken from $(x_0,\ldots,x_m)$ and no two adjacent letters are the same.
To give an example, if $m=2$ and our word is $w(x_0,x_1,x_2)=x_0^2 x_1^3x_0^4$ we have 
\begin{align}
    x_0^2 x_1^3x_0^4 =x_0^{n_0} x_1^{n_1} x_0^{n_3},
\end{align}
so that the word has $L=3$ letters. 

\begin{definition}
	Given a word $w$ of length $L$ and a tuple of $m$ positive semi-definite matrices $(A_0,\ldots,A_{m})$ we define a function $f_w(s|A_0,\ldots,A_{m}):\mathbb R^L \rightarrow \mathbb C$ by replacing each exponent $n_j$ with a real variable $s_j$ in the word $w(A_0,\ldots,A_m)$ and taking the trace.
\end{definition}

In the definition, we need to specify how we take real powers of positive semi-definite matrices, which is done by functional calculus (with $0^0=0$). Note that well-definedness of $s\mapsto X^s$ for a normal matrix $X$ requires positivity of all eigenvalues, i.e., positive semi-definiteness of $X$.

If we consider our previous example $w(x_0,x_1,x_2)=x_0^0 x_2^2 x_1^1 x_0^4$, then the function we now consider is given by
\begin{align}
	\RR^4\ni s \mapsto f_w(s|A_0,A_1,A_2) = \Tr[A_0^{s_1} A_2^{s_2} A_1^{s_3} A_0^{s_4}].
\end{align}
Note that a word of $k$ variables may only involve a smaller number of variables. For example $w(x_0,x_1,x_2)=x_1$. 
We call the variables that actually appear in a word the \emph{participating} variables. 

\begin{lemma}\label{lemma}
	Let $w$ be a word of length $L$ and of $2(m+1)$ variables. Then the function $f_w(s|A_0,\ldots,A_{m})$ has the following properties:
	\begin{enumerate}
		\item\label{prop:factorization} Factorization: $f_w(s|A_0\otimes C_0,\ldots,A_{m}\otimes C_{m  }) = f_w(s|A_0,\ldots,A_{m})f_w(s|C_0,\ldots,C_{m})$
		\item\label{prop:analyticity}  Analyticity: $s\mapsto f_w(s|A_0,\ldots,A_{m})$ is analytic. 
		\item\label{prop:normalization} $f_w(0,\ldots,0|A_0,\ldots,A_{m})\geq \mathrm{rank}(A_0)$ if $A_i$ has full rank for $i\geq 1$. 
	\end{enumerate}
\end{lemma}
\begin{proof}
	The first property (factorization) follows from $\tr(X\otimes Y) = \tr X \, \tr Y$ and from $(X\otimes Y)^s = X^s\otimes Y^s$, which only makes sense for positive $X$ and $Y$.
	Analyticity follows from the fact that $s\mapsto A^s$ is an analytic function (in contrast to the function $A\mapsto A^s)$. Writing out $f_w(s|A_0,\ldots,A_m)$ in terms of the spectral decompositions of the involved matrices shows that it is a 
	finite sum of finite products of analytic functions and hence analytic. 
	For the last property, there are two cases: If $A_0$ participates in the word $w(A_0,\ldots,A_m)$, then $f_w(0,\ldots,0|A_0,\ldots,A_m) = \mathrm{rank}(A_0)$, since $A_i^0=\id$ for $i\geq 1$ and $A_0^0=P_{\mathrm{supp}(A_0)}$.
	Otherwise, $f_w(0,\ldots,0|A_0,\ldots,A_m)$ coincides with the dimension $d$ of the matrices since all participating matrices have full rank.
\end{proof}

Now fix an arbitrary word $w$ of length $L$. By taking appropriate (real) powers of the eqs. $U A^{(i)}\otimes C^{(i)} U^\dagger=B^{(i)}\otimes C^{(i)}$, multiplying them appropriately and taking the trace, we find
\begin{align}
	f_w&(s|A^{(0)}\otimes C^{(0)},\ldots,A^{(m)}\otimes C^{(m)}) \nonumber\\
	&= f_w(s| B^{(0)}\otimes C^{(0)},\ldots,B^{(m)}\otimes C^{(m)}). 
\end{align}
By the factorization property in Lemma~\ref{lemma}, we find
\begin{align}
	&\left[	f_w(s|A^{(0)},\ldots, A^{(m)})-	f_w(s|B^{(0)},\ldots,B^{(m)})\right] \nonumber\\
	&\quad\quad\quad\quad\times f_w(s|C^{(0)},\ldots,C^{(m)})=0.
\end{align}
Since all involved functions are analytic by Lemma~\ref{lemma}, we have a product of two analytic functions that vanishes identically.
This implies that at least one of the factors needs to vanish identically.
But by property~\ref{prop:normalization} in Lemma~\ref{lemma} and the requirement that $C^{(i)}$ has full rank for $i\geq 1$, we know that $f_w(s|C^{(0)},\ldots,C^{(m)})$ cannot vanish identically.
Therefore 
\begin{align}\label{eq:f-conclusion}
	f_w(s|A^{(0)},\ldots,A^{(m)})=f_w(s|B^{(0)},\ldots,B^{(m)})
\end{align}
for all $s\in \mathbb R^L$.
Hence, in particular, we have
\begin{equation}
\begin{aligned}
	\Tr[w(A^{(0)},\ldots,A^{(m)})] = \Tr[w(B^{(0)},\ldots,B^{(m)})].
\end{aligned}
\end{equation}
which we wanted to show. 

We close this section with a remark about Theorem~\ref{prop:main}. One may wonder whether the condition that the $C^{(i)}$ have full rank for $i\geq 1$ is necessary.
Indeed at least some requirement on the support of the matrices is necessary. For example, let $A^{(0)} = \sigma_x, B^{(0)}= \sigma_y, A^{(1)} = \sigma_x, B^{(1)} = \sigma_z$.
Then for each $i=1,2$ there exist \emph{distinct} unitaries $U_i$ such that $U_i A^{(i)} U_i^\dagger = B^{(i)}$ but not a single unitary that works for both $i=1$ and $i=2$.
However, by choosing $C^{(i)}=\proj{i}$ and $U = \sum_i U_i \otimes \proj{i}$, all conditions of the proposition apart from the condition on the rank would be fulfilled. 
It is also not sufficient that all the $C^{(i)}$ have pairwise overlapping supports: One can come up with examples of three pairs of Hermitian matrices $(A^{(i)},B^{(i)})$ such that every \emph{two} pairs of the matrices are jointly unitarily equivalent, but not all three of them at once, and at the same time there exists three positive matrices with pairwise overlapping supports $C^{(i)}$ and a unitary $U$ such that $UA^{(i)}\otimes C^{(i)}U^\dagger = B^{(i)}\otimes C^{(i)}$ for $i=1,2,3$.
An example is given in the Appendix. 

\section{Good quantum reference frames degrade little} \label{sec:embezzlement} 
Theorem~\ref{thm:finite} implies that a quantum reference frame in a pure state that is used beneficially necessarily degrades by building up correlations. 
In this section, we show a counter-statement: if the quantum reference frame can be used to
implement approximately unitary dynamics on $S$, then it can be used in such a way that it hardly degrades 
(of course, one can always use it in a non-optimal way as well). 
Before coming to the Theorem, let us explain heuristically why we can expect this to be true. 
Label the reference frame by $C$ and suppose that the joint dynamics on $SC$ is unitary. 
Let the initial state of $C$ by $\ket\phi$ and consider two orthogonal states $\ket{\psi_1}$ and $\ket{\psi_2}$ on $S$. 
Then if the induced dynamics on $S$ is approximately given by the unitary $V$, we must have:
\begin{align}
	U\ket{\psi_1}\ket\phi &\approx (V\ket{\psi_1})\ket{\phi'_1}\\
	U\ket{\psi_2}\ket\phi &\approx (V\ket{\psi_2})\ket{\phi'_2}.
\end{align}
But by linearity, we must also have
\begin{align}
	U(\alpha \ket{\psi_1}+\beta \ket{\psi_2})\ket\phi &\approx \alpha (V\ket{\psi_1})\ket{\phi'_1}+\beta (V\ket{\psi_2})\ket{\phi'_2}\nonumber\\
	&\approx  (\alpha V\ket{\psi_1}+\beta V\ket{\psi_2})\ket{\phi'_3}.
\end{align}
This is only possible if $\ket{\phi'_1}\approx\ket{\phi'_2}\approx\ket{\phi'_3}$, i.e. the final state of $C$ only depends little on the initial state on $S$ (up to global phases).  
We see here a particular example of the \emph{information-disturbance trade-off} in quantum mechanics \cite{Fuchs1996,Fuchs1998,Kretschmann2008}:
If only little information can be gained about the initial state on $S$ by measuring $C$, then the disturbance must be small.
Conversely, if there is little disturbance on $S$, then only little information about the initial state of $S$ can be gained by measuring the final state of $C$. 

However, what the discussion so far hasn't shown is why the final state on $C$ is close to its initial state.
$C$ could be rotated unitarily as well. One could hope this unitary to be covariant so that it could be undone by a (unitary) covariant quantum channel.
But the whole point of using the reference frame is to be able to implement on $S$ a unitary that is \emph{not} covariant, so why should the effective unitary on $C$ be covariant?
Nevertheless, we will show that $C$ can be brought back close to its initial state using a covariant quantum channel.
The intuitive explanation for why this is the case is the following:
Since the unitary on $S$ is not covariant, it must fail to commute with at least some conserved quantity, meaning that the quantity is, in fact, not conserved under the dynamics induced by $U$. 
Implementing the unitary therefore requires a coherent source or sink with which the quantity can be exchanged.
However, our considerations above show that the final state on $S$ is essentially independent of how much of the conserved quantity is exchanged with $S$.  
This is possible only if the amount of the conserved quantity in $C$ is highly uncertain, to begin with (which in turn requires $C$ to be much larger than $S$ in the sense that it can store much larger amounts of the conserved quantity). 
But in this case, any unitary rotation induced by $S$ on $C$ must leave the state approximately invariant. 
This explanation leaves many open questions, in particular, if there are several non-commuting conserved quantities, but it provides a rough idea. 
Let us now come to our theorem. To state it, we use the \emph{diamond-norm} for quantum channels, which is defined as:
\begin{align}
	\norm{T_1-T_2}_\Diamond := \sup_{n\in\NN} \sup_\rho \norm{(T_1\otimes \Id_n)[\rho] - (T_2\otimes \Id_n)[\rho]}_1,\nonumber
\end{align}
where $\Id_n$ denotes the identity quantum channel on an $n$-dimensional quantum system and $\rho$ denote arbitrary quantum states on the joint system.
The diamond-norm quantifies the distinguishability of two channels on a given system taking into account possible correlations to other systems. 
\begin{theorem}\label{thm:qrf-back-action}
	Consider quantum systems $S$ and $C$ with projective unitary representations of some (arbitrary) group $G$, a covariant quantum channel $T$ on $SC$ and a state $\sigma_C$ on $C$ such that the induced channel on $S$
	\begin{align}
		T_S[\rho_S] = \tr_C[T[\rho_S\otimes\sigma_C]]
	\end{align}
	is close to unitary: $\norm{T_S-V(\cdot)V^\dagger}_\Diamond \leq \epsilon$. 
	Then there exists a covariant quantum channel $T'$ on $SC$ with identical induced dynamics on $S$, $T'_S=T_S$, but such that
	\begin{align}
	   D(\tr_S[T'[\rho_S\otimes\sigma_C]], \sigma_C) \leq 2\sqrt{2\epsilon},
	\end{align}
	for all states $\rho_S$ on $S$. 
 \end{theorem}

In contrast to Theorem~\ref{thm:finite}, this result does not transfer to the resource theory of athermality since even if $T$ is a thermal operation, we currently cannot guarantee that $T'$ also is.
We leave this as an open problem.

The fact that the system $C$ can be used in such a way that it changes little is reminiscent of the phenomenon of \emph{embezzlement} in the theory of entanglement \cite{van_dam2003universal}. There, a suitable resource state (the embezzler) may be used to implement, via local unitary operations, arbitrary state-transitions between pure states on a finite dimensional bipartite quantum system while experiencing arbitrary little back-action.
Embezzlers also exist in other resource theories, including asymmetry and athermality, see for example \cite{lipka-bartosik_catalysis_2023} for a comprehensive discussion.
From this point of view, Theorem~\ref{thm:qrf-back-action} shows that if a quantum reference frame is sufficiently good to implement arbitrary  unitaries to high precision, then it must be an embezzler. However, we emphasize that embezzlement is usually concerned with the ability to enable arbitrary state transitions with little back-action, whereas Theorem~\ref{thm:qrf-back-action} is concerned with implementing arbitrary \emph{unitary quantum channels} with little back-action.

A particularly interesting obstruction arising from conservation laws is the impossibility of measuring exactly (in the sense of von Neumann) an observable that does not commute with a conserved quantity. 
This statement is known as Wigner-Araki-Yanase (WAY) theorem \cite{Wigner1952,Araki1960,Yanase1961}, and has recently been re-interpreted and generalized in several ways (see, for example, \cite{Ozawa2002,Ozawa2002a,Busch2011,Marvian2012a,Miyadera2016,Ahmadi2013,Mohammady2021,Kuramochi2022}).
From the point of view of the resource theory of asymmetry, 
the theorem may be understood as saying that the measurement apparatus implementing an approximate measurement of an observable that does not commute with a conserved quantity must include a subsystem that serves as a quantum reference frame. 
Our theorem then implies that a measurement apparatus that can be used to implement such a measurement to very high accuracy can also be used so that its internal quantum reference frame hardly degrades. 
In particular, the measurement apparatus may be used again once its pointer is reset.

\subsection{The proof of theorem~\ref{thm:qrf-back-action}}
The proof consists of several steps:
First, we restrict to the case where $T$ is a unitary quantum channel, and $\sigma_C$ is pure.
Second, we use the quantitative version of the information-disturbance tradeoff from Ref.~\cite{Kretschmann2008} to argue that the final state on $C$ is almost independent of the input $\rho_S$ on $S$ and still approximately pure.
Third, since the final state on $C$ barely depends on $\rho_S$, we consider the situation in which $\rho_S=\id_S/d_S$ is maximally mixed. 
In this case, the dynamics on $C$ is covariant as well and doubly stochastic. 
We then show that when a doubly stochastic, covariant channel maps a pure state close to a pure state, there exists a covariant recovery map $R$ that maps the final state on $C$ close to the initial state $\sigma_C$. 
The channel $T'$ can now be chosen as $(I_S\otimes R)\circ T$ (with $I_S$ denoting the identity channel on $S$).
Finally, we use the covariant Stinespring theorem to generalize the result to the case where $T$ is not unitary, and $\sigma_C$ is mixed.

We now go through the steps in detail. So in the following, we assume that $T=U(\cdot)U^\dagger$ is a unitary, covariant quantum channel and $\sigma_C=\proj{\phi}$.
Refs.~\cite{Kretschmann2008,vomEnde} then shows that if $\norm{T_S-V(\cdot)V^\dagger}_\Diamond \leq \epsilon$ for some unitary $V$, then there exists a unitary $W$ such that
\begin{align}\label{eq:information-disturbance}
	\sup_{\norm{\ket\psi_S}=1} \norm{U \ket\psi_S\ket\phi_C - V\ket\psi_S W\ket\phi_C}^2\leq2\epsilon
\end{align}
We now make use of the fidelity of quantum states, defined as $F(\rho,\sigma) = \norm{\sqrt\rho\sqrt\sigma}_1=\Tr[\sqrt{\sqrt\rho \sigma \sqrt\rho}]$. 
It is non-decreasing under partial traces, so that $F(\rho_{SC},\sigma_{SC})\leq F(\rho_C,\sigma_C)$ for any two density matrices $\rho_{SC}$ and $\sigma_{SC}$. 
Moreover, for pure states, we have
\begin{align}
 	F(\Psi,\Phi) \ge 1- \frac12\norm{\ket\Psi-\ket\Phi}^2
\end{align}
where we use the short-hand notation $\Psi=\proj{\Psi}$. 
For a given state $\rho$ on $S$ let 
\begin{align}
	\hat T_\rho[\sigma_C] = \tr_S[U \rho\otimes\sigma_C U^\dagger]
\end{align}
be the final state on $C$. Then for a pure state $\rho=\proj\psi$ the information-disturbance tradeoff \eqref{eq:information-disturbance} implies 
\begin{align}
	F(\hat T_\psi[\sigma_C],W\sigma_C W^\dagger) \geq 1 - \epsilon.
\end{align}
However, if $\rho=\sum_i p_i \proj{\psi_i}$, then $\hat T_\rho[\sigma_C] = \sum_i p_i \hat T_{\psi_i}[\sigma_C]$. 
Since the fidelity is concave in its arguments, we, therefore, find 
\begin{align}\label{eq:fidelity-1}
	F(\hat T_\rho[\sigma_C],W\sigma_C W^\dagger) \geq 1-\epsilon
\end{align}
for all density matrices $\rho$ on $S$.
We now make use of a simple observation regarding the fidelity and so-called doubly stochastic quantum channels, i.e.\ quantum channels $T$ such that $T[\id]=\id$.
For a doubly stochastic quantum channel, its Hilbert-Schmidt dual $T^*$, defined by $\tr[AT[B]] = \tr[T^*[A]B]$ for all $A,B$, is also a quantum channel. 
Moreover, if $T$ is covariant, then so is $T^*$.
Now, if $\Psi=\proj\Psi$ and $\sigma$ is an arbitrary density matrix, then $F(\Psi,\sigma)^2 = \bra\Psi\sigma\ket\Psi=\Tr[\Psi\sigma]$. 
Using the definition of the Hilbert-Schmidt dual, we thus get
		\begin{align}
		F(T^*[\Psi],\Phi)^2 = \tr[T^*[\Psi]\Phi] = \tr[\Psi T[\Phi]] = F(\Psi,T[\Phi])^2\nonumber
		\end{align}
for any two pure states $\Psi=\proj\Psi$ and $\Phi=\proj\Phi$. 
The channel $\hat T_{\id_S/d_S}$ is doubly stochastic because it has a dilation in terms of a maximally mixed state. 
Therefore its dual $R:=(\hat T_{\id_S/d_S})^*$ is a covariant quantum channel and fulfills
\begin{align}\label{eq:fidelity-recovery}
	F(\sigma_C,R[W\sigma_C W^\dagger]) &= F(\hat T_{\id_S/d_S}[\sigma_C],W\sigma_C W^\dagger) \nonumber\\
	&\geq 1-\epsilon.
\end{align}	
In the following, we make use of the Fuchs-van de Graaf inequality \cite{Fuchs1999}
\begin{align}
	D(\rho,\sigma)\leq \sqrt{1-F(\rho,\sigma)^2}.
\end{align}
Together with \eqref{eq:fidelity-1} and \eqref{eq:fidelity-recovery}, it implies
\begin{align}
	D(\hat T_\rho[\sigma_C],W\sigma_C W^\dagger) \leq \sqrt{ 2\epsilon}
\end{align}
as well as
\begin{align}
	D(\sigma_C ,R[W\sigma_C W^\dagger])\leq \sqrt{2\epsilon}.
\end{align}
Then using a triangle-inequality and the data-processing inequality for the trace-distance yields
\begin{align}
	D(R\circ \hat T_\rho[\sigma_C],\sigma_C) &\leq D(R\circ \hat T_\rho[\sigma_C],R[W\sigma_C W^\dagger]) \nonumber\\
	&\quad+ D(R[W\sigma_C W^\dagger],\sigma_C)\\
	&\leq D(\hat T_\rho[\sigma_C],W\sigma_C W^\dagger) \nonumber\\
	&\quad +D(\sigma_C, R[W\sigma_C W^\dagger])\\
	&\leq 2\sqrt{2\epsilon}.
\end{align}
This finishes the proof in the case of unitary $T$ and pure $\sigma_C$ by setting $T'=(\Id_S\otimes R)\circ T$.
We now use this result to prove the general case, where $T$ need not be unitary and $\sigma_C$ need not be pure.
By the covariant Stinespring theorem, there exists a symmetric pure state $\omega_E = \proj{\chi}_E$ and a covariant \emph{unitary} channel $\tilde T$ on $SCE$ such that
\begin{align}
	T[\rho_{SC}] = \tr_E[\tilde T[\rho_{SC}\otimes\omega_E]].
\end{align}
Now consider a purification $\ket\Phi_{CC'}$ of $\sigma_C$ on a system $C'$. Then the channel $\tilde T\otimes \Id_{C'}$ is a covariant unitary channel on $SCEC'$ implementing $V(\cdot)V^\dagger$ on $S$ using a pure reference frame on $CEC'$ in the state $\proj{\Phi}_{CC'}\otimes \proj{\chi}_E$.
Since $\tilde T\otimes \Id_{C'}$ does not act on $C'$ and since this property is preserved when going to the Hilbert-Schmidt dual, by the previous reasoning, there is a covariant channel $\tilde R = \tilde R_{CE}\otimes \Id_{C'}$ on $CEC'$ such that the covariant quantum channel
\begin{align}
	\tilde T' &= (I_S\otimes \tilde R)\circ (\tilde T\otimes \Id_{C'}) = [(\Id_S\otimes \tilde R_{CE})\circ \tilde T]\otimes \Id_{C'}\nonumber\\
	&= \tilde T'_{SCE}\otimes \Id_{C'}
\end{align} 
leaves $CEC'$ approximately invariant and does not act on $C'$ while implementing $V(\cdot)V^\dagger$ with the same precision.
Because trace-distance is non-increasing under partial traces, the state on $C$ is also approximately invariant with at most the same error.
We can therefore define $T'$ by
\begin{align}
	T'[\rho_{SC}] = \tr_E\!\big[\tilde T'_{SCE}[\rho_{SC}\otimes\omega_E]\big],
\end{align}
which is covariant because $\omega_E$ is symmetric and $\tilde T'_{SCE}$ is covariant.
This finishes the proof in the general case.

\section{Discussion and outlook}
Our first result shows that a catalyst in the context of the resource theories of asymmetry and athermality must become correlated to some other degrees of freedom to be of any use.
The amount of correlations measured in terms of mutual information is precisely the change of entropy of system and environment combined, but our results
do not say anything about how these correlations can be distributed.
It is known in the resource theory of athermality that it can be quite beneficial to allow the catalyst to become (even only slightly) correlated to $S$, see Refs.~\cite{Gallego2016,Wilming2017a,Mueller2018,Rethinasamy2020,Shiraishi2021}. 
But do the same results hold if one bounds also the correlations with the thermal environment $E$, or is there a tradeoff, so that correlations to $E$ have to increase as those to $S$ decrease?

Our second result shows that a good quantum reference frame can be used in such a way that its state hardly changes. In particular, the result shows that a quantum reference frame that can implement any unitary transformation to high accuracy on a given system must automatically be an embezzler in the sense that it is almost catalytic \cite{Dam2003}. 
We have left the corresponding statement for the resource theory of athermality as an open problem, see also \cite{Woods2019} for related considerations.

Here, we have restricted ourselves to systems described by finite-dimensional Hilbert spaces.
It would be interesting to know whether similar results can be derived for infinite dimensional systems. 
This would allow for the treatment of unitary representations of non-compact and non-abelian groups, such as the Galilean or Poincaré group, which are relevant for treating reference frames for space-time localization and orientation.   

\emph{Note added.} Results closely related to Theorem~\ref{thm:qrf-back-action}, also using the quantitative information-disturbance tradeoff as a proof-technique,  have been previously derived in the context of so-called "no-programming theorems" for quantum computers \cite{Yang_2020,Gschwendtner_2021} and for estimating the fundamental energy costs of implementing unitary gates on a quantum computer \cite{Chiribella_2021,Yang_2022} (where the energy-source corresponds to the quantum reference frame in our formulation).
We thank Yuxiang Yang for making us aware of the close resemblance between these results and ours after the first version of this manuscript appeared as a preprint.

\emph{Acknowledgments. } We would like to thank Patryk Lipka-Bartosik, Victor Gitton, Nelly H. Y. Ng, Alexander Sottmeister, and Albert Werner for interesting discussions.
Support by the DFG through SFB 1227 (DQ-mat), Quantum Valley Lower Saxony, and funding by the Deutsche Forschungsgemeinschaft (DFG, German Research Foundation) under Germanys Excellence Strategy EXC-2123 QuantumFrontiers 390837967 is also acknowledged.

\bibliography{literature}

\begin{thebibliography}{74}
\providecommand{\natexlab}[1]{#1}
\providecommand{\url}[1]{\texttt{#1}}
\expandafter\ifx\csname urlstyle\endcsname\relax
  \providecommand{\doi}[1]{doi: #1}\else
  \providecommand{\doi}{doi: \begingroup \urlstyle{rm}\Url}\fi

\bibitem[Ahmadi et~al.(2010)Ahmadi, Jennings, and Rudolph]{Ahmadi2010}
M.~Ahmadi, D.~Jennings, and T.~Rudolph.
\newblock Dynamics of a quantum reference frame undergoing selective
  measurements and coherent interactions.
\newblock \emph{Phys. Rev. A}, 82\penalty0 (3):\penalty0 032320, sep 2010.
\newblock \doi{10.1103/physreva.82.032320}.

\bibitem[Ahmadi et~al.(2013)Ahmadi, Jennings, and Rudolph]{Ahmadi2013}
M.~Ahmadi, D.~Jennings, and T.~Rudolph.
\newblock The {Wigner}-{Araki}-{Yanase} theorem and the quantum resource theory
  of asymmetry.
\newblock \emph{New J. Phys.}, 15\penalty0 (1):\penalty0 013057, jan 2013.
\newblock \doi{10.1088/1367-2630/15/1/013057}.

\bibitem[Alexander et~al.(2022)Alexander, Gvirtz-Chen, and
  Jennings]{Alexander2022}
R.~Alexander, S.~Gvirtz-Chen, and D.~Jennings.
\newblock Infinitesimal reference frames suffice to determine the asymmetry
  properties of a quantum system.
\newblock \emph{New J. Phys.}, 24\penalty0 (5):\penalty0 053023, may 2022.
\newblock \doi{10.1088/1367-2630/ac688b}.

\bibitem[Araki and Yanase(1960)]{Araki1960}
H.~Araki and M.~M. Yanase.
\newblock Measurement of quantum mechanical operators.
\newblock \emph{Phys Rev}, 120\penalty0 (2):\penalty0 622--626, oct 1960.
\newblock \doi{10.1103/physrev.120.622}.

\bibitem[articleha P.~Woods and Horodecki(2023)]{Woods2019}
articleha P.~Woods and M.~Horodecki.
\newblock Autonomous quantum devices: When are they realizable without
  additional thermodynamic costs?
\newblock \emph{Physical Review X}, 13\penalty0 (1), feb 2023.
\newblock \doi{10.1103/physrevx.13.011016}.

\bibitem[Bargmann(1954)]{bargmann}
V.~Bargmann.
\newblock On unitary ray representations of continuous groups.
\newblock \emph{Annals of Mathematics}, pages 1--46, 1954.
\newblock \doi{10.2307/1969831}.

\bibitem[Bartlett et~al.(2006)Bartlett, Rudolph, Spekkens, and
  Turner]{Bartlett2006}
S.~D. Bartlett, T.~Rudolph, R.~W. Spekkens, and P.~S. Turner.
\newblock Degradation of a quantum reference frame.
\newblock \emph{New J. Phys.}, 8\penalty0 (4):\penalty0 58--58, apr 2006.
\newblock \doi{10.1088/1367-2630/8/4/058}.

\bibitem[Bartlett et~al.(2007{\natexlab{a}})Bartlett, Rudolph, Sanders, and
  Turner]{Bartlett2007a}
S.~D. Bartlett, T.~Rudolph, B.~C. Sanders, and P.~S. Turner.
\newblock Degradation of a quantum directional reference frame as a random
  walk.
\newblock \emph{J. Modern Opt.}, 54\penalty0 (13-15):\penalty0 2211--2221, sep
  2007{\natexlab{a}}.
\newblock \doi{10.1080/09500340701289254}.

\bibitem[Bartlett et~al.(2007{\natexlab{b}})Bartlett, Rudolph, and
  Spekkens]{Bartlett2007}
S.~D. Bartlett, T.~Rudolph, and R.~W. Spekkens.
\newblock Reference frames, superselection rules, and quantum information.
\newblock \emph{Rev. Mod. Phys.}, 79:\penalty0 555--609, Apr
  2007{\natexlab{b}}.
\newblock \doi{10.1103/RevModPhys.79.555}.

\bibitem[Boes et~al.(2019)Boes, Eisert, Gallego, Mueller, and
  Wilming]{Boes2019}
P.~Boes, J.~Eisert, R.~Gallego, M.~P. Mueller, and H.~Wilming.
\newblock Von {Neumann} entropy from unitarity.
\newblock \emph{Phys. Rev. Lett.}, 122\penalty0 (21):\penalty0 210402, May
  2019.
\newblock ISSN 0031-9007, 1079-7114.
\newblock \doi{10.1103/PhysRevLett.122.210402}.

\bibitem[Brandao et~al.(2013)Brandao, Horodecki, Oppenheim, Renes, and
  Spekkens]{Brandao2013}
F.~G. S.~L. Brandao, M.~Horodecki, J.~Oppenheim, J.~M. Renes, and R.~W.
  Spekkens.
\newblock The resource theory of quantum states out of thermal equilibrium.
\newblock \emph{Phys. Rev. Lett.}, 111:\penalty0 250404, 2013.
\newblock \doi{10.1103/PhysRevLett.111.250404}.

\bibitem[Brandao et~al.(2015)Brandao, Horodecki, Ng, Oppenheim, and
  Wehner]{Brandao2015}
F.~G. S.~L. Brandao, M.~Horodecki, N.~H.~Y. Ng, J.~Oppenheim, and S.~Wehner.
\newblock The second laws of quantum thermodynamics.
\newblock \emph{PNAS}, 112:\penalty0 3275--3279, 2015.
\newblock \doi{10.1073/pnas.1411728112}.

\bibitem[Busch and Loveridge(2011)]{Busch2011}
P.~Busch and L.~Loveridge.
\newblock Position measurements obeying momentum conservation.
\newblock \emph{Phys. Rev. Lett.}, 106\penalty0 (11):\penalty0 110406, mar
  2011.
\newblock \doi{10.1103/physrevlett.106.110406}.

\bibitem[Chiribella et~al.(2021)Chiribella, Yang, and Renner]{Chiribella_2021}
G.~Chiribella, Y.~Yang, and R.~Renner.
\newblock Fundamental energy requirement of reversible quantum operations.
\newblock \emph{Physical Review X}, 11\penalty0 (2), apr 2021.
\newblock \doi{10.1103/physrevx.11.021014}.

\bibitem[Ding et~al.(2021)Ding, Hu, and Fan]{Ding2021}
F.~Ding, X.~Hu, and H.~Fan.
\newblock Amplifying asymmetry with correlating catalysts.
\newblock \emph{Phys. Rev. A}, 103\penalty0 (2):\penalty0 022403, Feb. 2021.
\newblock ISSN 2469-9926, 2469-9934.
\newblock \doi{10.1103/PhysRevA.103.022403}.

\bibitem[Eisert and Wilkens(2000)]{Eisert2000}
J.~Eisert and M.~Wilkens.
\newblock Catalysis of {Entanglement} {Manipulation} for {Mixed} {States}.
\newblock \emph{Phys. Rev. Lett.}, 85\penalty0 (2):\penalty0 437--440, July
  2000.
\newblock ISSN 0031-9007, 1079-7114.
\newblock \doi{10.1103/PhysRevLett.85.437}.

\bibitem[Faist et~al.(2015)Faist, Dupuis, Oppenheim, and Renner]{Faist2015}
P.~Faist, F.~Dupuis, J.~Oppenheim, and R.~Renner.
\newblock The minimal work cost of information processing.
\newblock \emph{Nature Comm.}, 6:\penalty0 7669, 2015.
\newblock \doi{10.1038/ncomms8669}.

\bibitem[Fuchs and van~de Graaf(1999)]{Fuchs1999}
C.~Fuchs and J.~van~de Graaf.
\newblock Cryptographic distinguishability measures for quantum-mechanical
  states.
\newblock \emph{{IEEE} Transactions on Information Theory}, 45\penalty0
  (4):\penalty0 1216--1227, may 1999.
\newblock \doi{10.1109/18.761271}.

\bibitem[Fuchs(1998)]{Fuchs1998}
C.~A. Fuchs.
\newblock Information gain vs. state disturbance in quantum theory.
\newblock \emph{Fortschr. Phys.}, 46\penalty0 (4-5):\penalty0 535--565, 1998.
\newblock
  \doi{10.1002/(SICI)1521-3978(199806)46:4/5<535::AID-PROP535>3.0.CO;2-0}.

\bibitem[Fuchs and Peres(1996)]{Fuchs1996}
C.~A. Fuchs and A.~Peres.
\newblock Quantum-state disturbance versus information gain: Uncertainty
  relations for quantum information.
\newblock \emph{Phys. Rev. A}, 53\penalty0 (4):\penalty0 2038--2045, apr 1996.
\newblock \doi{10.1103/physreva.53.2038}.

\bibitem[Gallego et~al.(2016)Gallego, Eisert, and Wilming]{Gallego2016}
R.~Gallego, J.~Eisert, and H.~Wilming.
\newblock Thermodynamic work from operational principles.
\newblock \emph{New J. Phys.}, 18\penalty0 (10):\penalty0 103017, 2016.
\newblock \doi{10.1088/1367-2630/18/10/103017}.

\bibitem[Gour and Spekkens(2008)]{Gour2008}
G.~Gour and R.~W. Spekkens.
\newblock The resource theory of quantum reference frames: manipulations and
  monotones.
\newblock \emph{New J. Phys.}, 10\penalty0 (3):\penalty0 033023, mar 2008.
\newblock \doi{10.1088/1367-2630/10/3/033023}.

\bibitem[Gour et~al.(2009)Gour, Marvian, and Spekkens]{Gour2009}
G.~Gour, I.~Marvian, and R.~W. Spekkens.
\newblock Measuring the quality of a quantum reference frame: The relative
  entropy of frameness.
\newblock \emph{Phys. Rev. A}, 80\penalty0 (1):\penalty0 012307, jul 2009.
\newblock \doi{10.1103/physreva.80.012307}.

\bibitem[Gour et~al.(2015)Gour, Müller, Narasimhachar, Spekkens, and
  Halpern]{Gour2015}
G.~Gour, M.~P. Müller, V.~Narasimhachar, R.~W. Spekkens, and N.~Y. Halpern.
\newblock The resource theory of informational nonequilibrium in
  thermodynamics.
\newblock \emph{Phys. Rep.}, 583:\penalty0 1--58, jul 2015.
\newblock \doi{10.1016/j.physrep.2015.04.003}.

\bibitem[Gour et~al.(2018)Gour, Jennings, Buscemi, Duan, and Marvian]{Gour2018}
G.~Gour, D.~Jennings, F.~Buscemi, R.~Duan, and I.~Marvian.
\newblock Quantum majorization and a complete set of entropic conditions for
  quantum thermodynamics.
\newblock \emph{Nat Commun}, 9\penalty0 (1):\penalty0 5352, Dec. 2018.
\newblock ISSN 2041-1723.
\newblock \doi{10.1038/s41467-018-06261-7}.

\bibitem[Gschwendtner et~al.(2021)Gschwendtner, Bluhm, and
  Winter]{Gschwendtner_2021}
M.~Gschwendtner, A.~Bluhm, and A.~Winter.
\newblock Programmability of covariant quantum channels.
\newblock \emph{Quantum}, 5:\penalty0 488, jun 2021.
\newblock \doi{10.22331/q-2021-06-29-488}.

\bibitem[Horodecki and Oppenheim(2013)]{Horodecki2013}
M.~Horodecki and J.~Oppenheim.
\newblock Fundamental limitations for quantum and nanoscale thermodynamics.
\newblock \emph{Nature Comm.}, 4:\penalty0 2059, 2013.
\newblock \doi{10.1038/ncomms3059}.

\bibitem[Janzing(2006)]{Janzing2006}
D.~Janzing.
\newblock Quantum thermodynamics with missing reference frames: Decompositions
  of free energy into non-increasing components.
\newblock \emph{J. Stat. Phys.}, 125\penalty0 (3):\penalty0 761--776, nov 2006.
\newblock \doi{10.1007/s10955-006-9220-x}.

\bibitem[Janzing et~al.(2000)Janzing, Wocjan, Zeier, Geiss, and
  Beth]{Janzing2000}
D.~Janzing, P.~Wocjan, R.~Zeier, R.~Geiss, and T.~Beth.
\newblock Thermodynamic cost of reliability and low temperatures: Tightening
  landauer's principle and the second law.
\newblock \emph{Int. J. Th. Phys.}, 39:\penalty0 2717, 2000.
\newblock \doi{10.1023/A:1026422630734}.

\bibitem[Jonathan and Plenio(1999)]{Jonathan1999}
D.~Jonathan and M.~B. Plenio.
\newblock Entanglement-{Assisted} {Local} {Manipulation} of {Pure} {Quantum}
  {States}.
\newblock \emph{Phys. Rev. Lett.}, 83\penalty0 (17):\penalty0 3566--3569, Oct.
  1999.
\newblock ISSN 0031-9007, 1079-7114.
\newblock \doi{10.1103/PhysRevLett.83.3566}.

\bibitem[Keyl and Werner(1999)]{Keyl1999}
M.~Keyl and R.~F. Werner.
\newblock Optimal cloning of pure states, testing single clones.
\newblock \emph{J. Math. Phys.}, 40\penalty0 (7):\penalty0 3283--3299, jul
  1999.
\newblock \doi{10.1063/1.532887}.

\bibitem[Kondra et~al.(2021)Kondra, Datta, and Streltsov]{Kondra2021}
T.~V. Kondra, C.~Datta, and A.~Streltsov.
\newblock Catalytic transformations of pure entangled states.
\newblock \emph{Physical Review Letters}, 127\penalty0 (15):\penalty0 150503,
  oct 2021.
\newblock \doi{10.1103/physrevlett.127.150503}.

\bibitem[Kretschmann et~al.(2008)Kretschmann, Schlingemann, and
  Werner]{Kretschmann2008}
D.~Kretschmann, D.~Schlingemann, and R.~F. Werner.
\newblock The information-disturbance tradeoff and the continuity of
  stinespring{\textquotesingle}s representation.
\newblock \emph{{IEEE} Transactions on Information Theory}, 54\penalty0
  (4):\penalty0 1708--1717, apr 2008.
\newblock \doi{10.1109/tit.2008.917696}.

\bibitem[Kuramochi and Tajima(2022)]{Kuramochi2022}
Y.~Kuramochi and H.~Tajima.
\newblock Wigner-araki-yanase theorem for continuous and unbounded conserved
  observables.
\newblock 2022.
\newblock \doi{10.48550/arxiv.2208.13494}.

\bibitem[Lipka-Bartosik and Skrzypczyk(2021)]{LipkaBartosik2021}
P.~Lipka-Bartosik and P.~Skrzypczyk.
\newblock Catalytic quantum teleportation.
\newblock \emph{Physical Review Letters}, 127:\penalty0 080502, Feb. 2021.
\newblock \doi{10.1103/PhysRevLett.127.080502}.

\bibitem[Lipka-Bartosik et~al.(2023{\natexlab{a}})Lipka-Bartosik,
  Perarnau-Llobet, and Brunner]{Patryk}
P.~Lipka-Bartosik, M.~Perarnau-Llobet, and N.~Brunner.
\newblock Operational definition of the temperature of a quantum state.
\newblock \emph{Physical Review Letters}, 130\penalty0 (4), jan
  2023{\natexlab{a}}.
\newblock \doi{10.1103/physrevlett.130.040401}.

\bibitem[Lipka-Bartosik et~al.(2023{\natexlab{b}})Lipka-Bartosik, Wilming, and
  Ng]{lipka-bartosik_catalysis_2023}
P.~Lipka-Bartosik, H.~Wilming, and N.~H.~Y. Ng.
\newblock Catalysis in quantum information theory.
\newblock 2023{\natexlab{b}}.
\newblock \doi{10.48550/arXiv.2306.00798}.

\bibitem[Lostaglio and Müller(2019)]{Lostaglio2019}
M.~Lostaglio and M.~P. Müller.
\newblock Coherence and {Asymmetry} {Cannot} be {Broadcast}.
\newblock \emph{Phys. Rev. Lett.}, 123\penalty0 (2):\penalty0 020403, July
  2019.
\newblock ISSN 0031-9007, 1079-7114.
\newblock \doi{10.1103/PhysRevLett.123.020403}.

\bibitem[Marvian(2022)]{Marvian2021}
I.~Marvian.
\newblock Operational interpretation of quantum fisher information in quantum
  thermodynamics.
\newblock \emph{Physical Review Letters}, 129\penalty0 (19), oct 2022.
\newblock \doi{10.1103/physrevlett.129.190502}.

\bibitem[Marvian and Spekkens(2012)]{Marvian2012a}
I.~Marvian and R.~W. Spekkens.
\newblock An information-theoretic account of the wigner-araki-yanase theorem.
\newblock 2012.
\newblock \doi{10.48550/arxiv.1212.3378}.

\bibitem[Marvian and Spekkens(2013)]{Marvian2013}
I.~Marvian and R.~W. Spekkens.
\newblock The theory of manipulations of pure state asymmetry: {I}. {Basic}
  tools, equivalence classes and single copy transformations.
\newblock \emph{New J. Phys.}, 15\penalty0 (3):\penalty0 033001, Mar. 2013.
\newblock ISSN 1367-2630.
\newblock \doi{10.1088/1367-2630/15/3/033001}.

\bibitem[Marvian and Spekkens(2016)]{Marvian2016}
I.~Marvian and R.~W. Spekkens.
\newblock How to quantify coherence: Distinguishing speakable and unspeakable
  notions.
\newblock \emph{Phys. Rev. A}, 94:\penalty0 052324, Nov 2016.
\newblock \doi{10.1103/PhysRevA.94.052324}.

\bibitem[Marvian and Spekkens(2019)]{Marvian2019}
I.~Marvian and R.~W. Spekkens.
\newblock A no-broadcasting theorem for quantum asymmetry and coherence and a
  trade-off relation for approximate broadcasting.
\newblock \emph{Phys. Rev. Lett.}, 123\penalty0 (2):\penalty0 020404, July
  2019.
\newblock ISSN 0031-9007, 1079-7114.
\newblock \doi{10.1103/PhysRevLett.123.020404}.

\bibitem[Marvian(2012)]{Marvian2012}
I.~M. Marvian.
\newblock \emph{Symmetry, Asymmetry and Quantum Information}.
\newblock PhD thesis, University of Waterloo, 2012.
\newblock URL \url{http://hdl.handle.net/10012/7088}.

\bibitem[Miyadera and Loveridge(2020)]{Miyadera2020}
T.~Miyadera and L.~Loveridge.
\newblock A quantum reference frame size-accuracy trade-off for quantum
  channels.
\newblock \emph{J. Phys.: Conf. Ser.}, 1638\penalty0 (1):\penalty0 012008, oct
  2020.
\newblock \doi{10.1088/1742-6596/1638/1/012008}.

\bibitem[Miyadera et~al.(2016)Miyadera, Loveridge, and Busch]{Miyadera2016}
T.~Miyadera, L.~Loveridge, and P.~Busch.
\newblock Approximating relational observables by absolute quantities: a
  quantum accuracy-size trade-off.
\newblock \emph{J. Phys. A: Math. Theor.}, 49\penalty0 (18):\penalty0 185301,
  mar 2016.
\newblock \doi{10.1088/1751-8113/49/18/185301}.

\bibitem[Mohammady et~al.(2023)Mohammady, Miyadera, and
  Loveridge]{Mohammady2021}
M.~H. Mohammady, T.~Miyadera, and L.~Loveridge.
\newblock Measurement disturbance and conservation laws in quantum mechanics.
\newblock \emph{Quantum}, 7:\penalty0 1033, jun 2023.
\newblock \doi{10.22331/q-2023-06-05-1033}.

\bibitem[Müller(2018)]{Mueller2018}
M.~P. Müller.
\newblock Correlating thermal machines and the second law at the nanoscale.
\newblock \emph{Phys. Rev. X}, 8\penalty0 (4):\penalty0 041051, dec 2018.
\newblock \doi{10.1103/physrevx.8.041051}.

\bibitem[Ozawa(2002{\natexlab{a}})]{Ozawa2002}
M.~Ozawa.
\newblock Conservative quantum computing.
\newblock \emph{Phys. Rev. Lett.}, 89\penalty0 (5):\penalty0 057902, jul
  2002{\natexlab{a}}.
\newblock \doi{10.1103/physrevlett.89.057902}.

\bibitem[Ozawa(2002{\natexlab{b}})]{Ozawa2002a}
M.~Ozawa.
\newblock Conservation laws, uncertainty relations, and quantum limits of
  measurements.
\newblock \emph{Phys. Rev. Lett.}, 88\penalty0 (5):\penalty0 050402, jan
  2002{\natexlab{b}}.
\newblock \doi{10.1103/physrevlett.88.050402}.

\bibitem[Poulin and Yard(2007)]{Poulin2007}
D.~Poulin and J.~Yard.
\newblock Dynamics of a quantum reference frame.
\newblock \emph{New J. Phys.}, 9\penalty0 (5):\penalty0 156--156, may 2007.
\newblock \doi{10.1088/1367-2630/9/5/156}.

\bibitem[Rethinasamy and Wilde(2020)]{Rethinasamy2020}
S.~Rethinasamy and M.~M. Wilde.
\newblock Relative entropy and catalytic relative majorization.
\newblock \emph{Phys. Rev. Research}, 2\penalty0 (3):\penalty0 033455, sep
  2020.
\newblock \doi{10.1103/physrevresearch.2.033455}.

\bibitem[Shapiro(1991)]{Shapiro1991}
H.~Shapiro.
\newblock A survey of canonical forms and invariants for unitary similarity.
\newblock \emph{Linear Algebra Appl.}, 147:\penalty0 101--167, mar 1991.
\newblock \doi{10.1016/0024-3795(91)90232-l}.

\bibitem[Shiraishi and Sagawa(2021)]{Shiraishi2021}
N.~Shiraishi and T.~Sagawa.
\newblock Quantum thermodynamics of correlated-catalytic state conversion at
  small scale.
\newblock \emph{Phys. Rev. Lett.}, 126\penalty0 (15):\penalty0 150502, apr
  2021.
\newblock \doi{10.1103/physrevlett.126.150502}.

\bibitem[Specht(1940)]{Specht1940}
W.~Specht.
\newblock Zur theorie der matrizen. ii.
\newblock \emph{Jahresber. Dtsch. Math.-Ver.}, 50:\penalty0 19--23, 1940.
\newblock URL \url{http://eudml.org/doc/146243}.

\bibitem[Tajima and Saito(2021)]{Tajima2021}
H.~Tajima and K.~Saito.
\newblock Universal limitation of quantum information recovery: symmetry versus
  coherence.
\newblock 2021.
\newblock \doi{https://doi.org/10.48550/arXiv.2103.01876}.

\bibitem[Tajima et~al.(2018)Tajima, Shiraishi, and Saito]{Tajima2018}
H.~Tajima, N.~Shiraishi, and K.~Saito.
\newblock Uncertainty relations in implementation of unitary operations.
\newblock \emph{Phys. Rev. Lett.}, 121\penalty0 (11):\penalty0 110403, sep
  2018.
\newblock \doi{10.1103/physrevlett.121.110403}.

\bibitem[Tajima et~al.(2020)Tajima, Shiraishi, and Saito]{Tajima2020}
H.~Tajima, N.~Shiraishi, and K.~Saito.
\newblock Coherence cost for violating conservation laws.
\newblock \emph{Phys. Rev. Research}, 2\penalty0 (4):\penalty0 043374, dec
  2020.
\newblock \doi{10.1103/physrevresearch.2.043374}.

\bibitem[Tajima et~al.(2022)Tajima, Takagi, and Kuramochi]{Tajima2022}
H.~Tajima, R.~Takagi, and Y.~Kuramochi.
\newblock Universal trade-off structure between symmetry, irreversibility, and
  quantum coherence in quantum processes.
\newblock 2022.
\newblock \doi{10.48550/arxiv.2206.11086}.

\bibitem[Vaccaro et~al.(2008)Vaccaro, Anselmi, Wiseman, and
  Jacobs]{Vaccaro2008}
J.~A. Vaccaro, F.~Anselmi, H.~M. Wiseman, and K.~Jacobs.
\newblock Tradeoff between extractable mechanical work, accessible
  entanglement, and ability to act as a reference system, under arbitrary
  superselection rules.
\newblock \emph{Phys. Rev. A}, 77:\penalty0 032114, Mar 2008.
\newblock \doi{10.1103/PhysRevA.77.032114}.

\bibitem[Vaccaro et~al.(2018)Vaccaro, Croke, and Barnett]{Vaccaro2018}
J.~A. Vaccaro, S.~Croke, and S.~M. Barnett.
\newblock Is coherence catalytic?
\newblock \emph{J. Phys. A: Math. Theor.}, 51\penalty0 (41):\penalty0 414008,
  Oct. 2018.
\newblock ISSN 1751-8113, 1751-8121.
\newblock \doi{10.1088/1751-8121/aac112}.

\bibitem[van Dam and Hayden(2003{\natexlab{a}})]{Dam2003}
W.~van Dam and P.~Hayden.
\newblock Universal entanglement transformations without communication.
\newblock \emph{Phys. Rev. A}, 67\penalty0 (6):\penalty0 060302, June
  2003{\natexlab{a}}.
\newblock ISSN 1050-2947, 1094-1622.
\newblock \doi{10.1103/PhysRevA.67.060302}.

\bibitem[van Dam and Hayden(2003{\natexlab{b}})]{van_dam2003universal}
W.~van Dam and P.~Hayden.
\newblock Universal entanglement transformations without communication.
\newblock \emph{Physical Review A}, 67\penalty0 (6):\penalty0 060302, June
  2003{\natexlab{b}}.
\newblock \doi{10.1103/PhysRevA.67.060302}.
\newblock Publisher: American Physical Society.

\bibitem[vom Ende(2023)]{vomEnde}
F.~vom Ende.
\newblock Progress on the kretschmann-schlingemann-werner conjecture.
\newblock 2023.
\newblock \doi{10.48550/arXiv.2308.15389}.

\bibitem[Wiegmann(1961)]{Wiegmann1961}
N.~A. Wiegmann.
\newblock Necessary and sufficient conditions for unitary similarity.
\newblock \emph{J. Aust. Math. Soc.}, 2\penalty0 (1):\penalty0 122--126, apr
  1961.
\newblock \doi{10.1017/s1446788700026422}.

\bibitem[Wigner(1952)]{Wigner1952}
E.~P. Wigner.
\newblock Die messung quantenmechanischer operatoren.
\newblock \emph{Zeitschrift für Physik A Hadrons and nuclei}, 133\penalty0
  (1-2):\penalty0 101--108, sep 1952.
\newblock \doi{10.1007/bf01948686}.

\bibitem[Wilming(2021)]{Wilming2021}
H.~Wilming.
\newblock Entropy and reversible catalysis.
\newblock \emph{Phys. Rev. Lett.}, 127:\penalty0 260402, Dec. 2021.
\newblock \doi{10.1103/PhysRevLett.127.260402}.

\bibitem[Wilming(2022)]{Wilming2022}
H.~Wilming.
\newblock Correlations in typicality and an affirmative solution to the exact
  catalytic entropy conjecture.
\newblock \emph{Quantum}, 6:\penalty0 858, nov 2022.
\newblock \doi{10.22331/q-2022-11-10-858}.

\bibitem[Wilming et~al.(2017)Wilming, Gallego, and Eisert]{Wilming2017a}
H.~Wilming, R.~Gallego, and J.~Eisert.
\newblock Axiomatic characterization of the quantum relative entropy and free
  energy.
\newblock \emph{Entropy}, 19\penalty0 (6):\penalty0 241, 2017.
\newblock \doi{10.3390/e19060241}.

\bibitem[Yanase(1961)]{Yanase1961}
M.~M. Yanase.
\newblock Optimal measuring apparatus.
\newblock \emph{Phys Rev}, 123\penalty0 (2):\penalty0 666--668, jul 1961.
\newblock \doi{10.1103/physrev.123.666}.

\bibitem[Yang et~al.(2020)Yang, Renner, and Chiribella]{Yang_2020}
Y.~Yang, R.~Renner, and G.~Chiribella.
\newblock Optimal universal programming of unitary gates.
\newblock \emph{Physical Review Letters}, 125\penalty0 (21), nov 2020.
\newblock \doi{10.1103/physrevlett.125.210501}.

\bibitem[Yang et~al.(2022)Yang, Renner, and Chiribella]{Yang_2022}
Y.~Yang, R.~Renner, and G.~Chiribella.
\newblock Energy requirement for implementing unitary gates on energy-unbounded
  systems.
\newblock \emph{Journal of Physics A: Mathematical and Theoretical},
  55\penalty0 (49):\penalty0 494003, dec 2022.
\newblock \doi{10.1088/1751-8121/ac717e}.

\bibitem[Yunger~Halpern and Renes(2016)]{YungerHalpern2016}
N.~Yunger~Halpern and J.~M. Renes.
\newblock Beyond heat baths: Generalized resource theories for small-scale
  thermodynamics.
\newblock \emph{Phys. Rev. E}, 93\penalty0 (2), Feb. 2016.
\newblock ISSN 2470-0053.
\newblock \doi{10.1103/physreve.93.022126}.

\bibitem[Åberg(2014)]{Aaberg2014}
J.~Åberg.
\newblock Catalytic {Coherence}.
\newblock \emph{Phys. Rev. Lett.}, 113\penalty0 (15):\penalty0 150402, Oct.
  2014.
\newblock ISSN 0031-9007, 1079-7114.
\newblock \doi{10.1103/PhysRevLett.113.150402}.

\end{thebibliography}
\appendix

\onecolumn
\section{The rank-condition in Proposition~\ref{prop:main}}

We here give the promised example of three pairs of Hermitian matrices $(A^{(i)},B^{(i)})$ such that every \emph{two} pairs of the matrices are jointly unitarily equivalent, but not all three of them at once, and at the same time there exist three positive matrices with pairwise overlapping supports $C^{(i)}$ and a unitary $U$ such that $UA^{(i)}\otimes C^{(i)}U^\dagger = B^{(i)}\otimes C^{(i)}$ for $i=1,2,3$.

We first describe a general way to find such examples and then give an explicit example.

We will choose $C^{(i)}$ as 
\begin{align}
    C^{(i)} = \proj{i} + \proj{i+1},
\end{align}
where addition is modulo $3$. Any product of the $C^{(i)}$ involving all the different $i$ is zero. Therefore, all words of the variables $A^{(i)}\otimes C^{(i)}$ where all the $i$ are participating are traceless.
Hence by Wiegmann's theorem, as long as every two pairs $(A^{(i)}, B^{(i)})$ are unitarily equivalent, then the $A^{(i)}\otimes C^{(i)}$ are jointly unitarily equivalent to the $B^{(i)}\otimes C^{(i)}$.

So we are looking for matrices such that there exists unitaries $U,V$ and $W$ with
\begin{align}
U A^{(i)} U^\dagger &= B^{(i)}, \quad i=1,2\\
V A^{(i)} V^\dagger &= B^{(i)}, \quad i=2,3\\
W A^{(i)} W^\dagger &= B^{(i)}, \quad i=3,1,
\end{align}
but we want that the Hermitian matrices $A^{(i)}$ and $B^{(i)}$ are not jointly unitarily equivalent for all $i=1,2,3$. For simplicity we choose $U=\id$ (this simply amounts to re-defining the $B^{(i)}$ so it does not limit the generality of what follows). 
The above equations imply that
\begin{align}
    A^{(1)} &= W A^{(1)} W^\dagger \\
    A^{(2)} &= V A^{(2)} V\\
    A^{(3)} &= V^\dagger W A^{(3)} W^\dagger V.
\end{align}
Let us define $\tilde U = V^\dagger W$, so that the equations read
\begin{align}
[A^{(1)}, W]=0,\ [A^{(2)},V]=0,\ [A^{(3)},\tilde U]=0
\end{align}
and
\begin{align}
    B^{(1)}=A^{(1)},\ B^{(2)}=A^{(2)},\ B^{(3)} = V A^{(3)} V^\dagger.
\end{align}
By Wiegmann's theorem, all we need now is two unitaries $V,W$ and three Hermitian matrices $A^{(i)}$ fulfilling the above conditions such that
\begin{align}
    \tr[B^{(1)}B^{(2)}B^{(3)}]&=\tr[A^{(1)}A^{(2)} V A^{(3)}V^\dagger]\nonumber\\ &\neq \tr[A^{(1)}A^{(2)} A^{(3)}].
\end{align}
An explicit such choice is as follows:
\begin{align}
A^{(1)}&=\begin{pmatrix}0 &0 &0\\
0&0 &0  \\
0&0 & 1\end{pmatrix},\quad
A^{(2)}=\begin{pmatrix}0 &0 &1\\
0&0 &0  \\
1&0 &0 \end{pmatrix},\quad
A^{(3)}=\mathrm{i}\begin{pmatrix}0 &\sqrt 3 &-\sqrt 3\\
-\sqrt 3&0 &\sqrt 3  \\
\sqrt 3&-\sqrt 3 &0 \end{pmatrix},\\
V &= \begin{pmatrix}0 &1 &0\\
1&0 &0  \\
0&0 &1 \end{pmatrix},\quad
W=\begin{pmatrix}0 &0 &1\\
0&1 &0  \\
1&0 &0 \end{pmatrix},
\end{align}
which yields
\begin{align}
    \left|\tr[B^{(1)}B^{(2)}B^{(3)}]-\tr[A^{(1)}A^{(2)} A^{(3)}]\right| = 2\sqrt{3}.
\end{align}

\end{document}